\documentclass[10pt,reqno]{amsart}
\usepackage[T1]{fontenc}
\usepackage[utf8]{inputenc}
\usepackage[english]{babel}
\usepackage{lmodern}
\usepackage{amsmath,amssymb,amsthm,mathtools}
\usepackage{booktabs,tabularx,array}
\usepackage{needspace}
\usepackage{microtype}
\usepackage[colorlinks=true,linkcolor=black,citecolor=black,urlcolor=blue]{hyperref}
\hypersetup{
 pdftitle={Construction and formal verification of a special topology on Z^n},
 pdfauthor={Ignacio Vega-Páez},
 pdfsubject={General theorem on Z^n and corollary for the integers},
 pdfkeywords={Marcus-Wyse topology, connectedness, Lean 4, Mathlib}
}
\usepackage{enumitem}
\setlist{nosep}

\newtheorem{theorem}{Theorem}
\newtheorem{proposition}[theorem]{Proposition}
\newtheorem{lemma}[theorem]{Lemma}
\newtheorem{corollary}[theorem]{Corollary}
\theoremstyle{definition}

\theoremstyle{remark}
\newtheorem{remark}[theorem]{Remark}

\title[Marcus--Wyse topology and formal verification]{Construction and formal verification of a special topology on $\mathbb Z^n$: characterization of connected subsets}
\author{Ignacio Vega-Páez}
\address{Centro de Investigación y de Estudios Avanzados del Instituto Politécnico Nacional (CINVESTAV-IPN), Mexico City, Mexico}
\email{ignaciovp633@gmail.com}
\thanks{Corresponding author: Ignacio Vega-Páez, \texttt{ignaciovp633@gmail.com}.}
\date{October 27, 2026}

\begin{document}
\begin{abstract}
We examine whether the integers admit a topology whose connected subsets are
precisely the sets of consecutive integers. The main result gives a
construction for every integer lattice $\mathbb Z^n$, $n\geq1$; the answer for
$\mathbb Z$ is then obtained as a corollary. We prove that the neighborhoods
determined by the parity of the sum of the coordinates form a basis and that
topological connectedness agrees exactly with accessibility by finite paths
of unit lattice steps. We also present a deductive audit and the
correspondence between this argument and a verification in Lean~4 and
Mathlib, including the empty set and the independent formal characterization
of integer intervals. The construction and mathematical result
originate in Problem 5712, proposed by Dan Marcus and solved in 1970 by the
Cleveland State University group; the present article contributes a
detailed exposition and a report on formal verification.
\end{abstract}

\maketitle

\section{Introduction and origin of the problem}

The motivating question is:
\begin{quote}
\emph{Is it possible to endow the set of integers with a topology in which
the connected subsets are precisely the sets of consecutive integers?}
\end{quote}

Its natural generalization replaces the integer line by $\mathbb Z^n$ and
consecutiveness by the existence of a finite path of adjacent points.
The question appeared as Problem 5712 in \emph{The American Mathematical
Monthly}. It was proposed by Dan Marcus, and its solution is attributed
to the Cleveland State University Problem Solving Group, with Frank Wyse
as adviser \cite{marcus1970}. The associated topology is commonly known
as the Marcus--Wyse topology.

Our purpose is not to claim a new solution to the historical question,
but to state it precisely, supply all the necessary logical steps, and
describe its subsequent formal verification in Lean~4/Mathlib. The main
statement is first formulated for $\mathbb Z^n$; the result for $\mathbb Z$
is then deduced as a special case.

\section{Definitions and main theorem}

Let $n\geq1$ and set $L_n=\mathbb Z^n$. For
$x=(x_1,\ldots,x_n)$ and $y=(y_1,\ldots,y_n)$ in $L_n$,
define \emph{unit adjacency} by
\begin{equation}\label{eq:adjacency}
 x\,R_n\,y
 \quad\Longleftrightarrow\quad
 d_1(x,y):=\sum_{k=1}^{n}|x_k-y_k|=1.
\end{equation}
Equivalently, $x$ and $y$ differ by one in exactly one coordinate
and agree in all the others. Define also
\begin{equation}\label{eq:parity}
 p(x)=\sum_{k=1}^{n}x_k.
\end{equation}

For $C\subseteq L_n$ and $x,y\in C$, write $xE_Cy$ and say
that $x$ and $y$ are \emph{accessible within $C$} if there are
$m\geq1$ and $z_1,\ldots,z_m\in C$ such that
\begin{equation}\label{eq:path}
 z_1=x,\qquad z_m=y,\qquad z_iR_nz_{i+1}\quad(1\leq i<m).
\end{equation}
A one-point path ($m=1$) is allowed when $x=y$.
We call $C$ \emph{consecutive in the lattice} if every pair of its
points is accessible in this way. The condition holds vacuously
for $C=\varnothing$.

For each $x\in L_n$, define
\begin{equation}\label{eq:basic}
 U_x=\begin{cases}
 \{x\}\cup\{y\in L_n:xR_ny\}, & p(x)\text{ is even},\\
 \{x\}, & p(x)\text{ is odd}.
 \end{cases}
\end{equation}

\begin{theorem}[Construction and characterization of the topology on $\mathbb Z^n$]
\label{thm:main}
For every integer $n\geq1$, the family
\[
 \mathcal B_n=\{U_x:x\in\mathbb Z^n\}
\]
is a basis\footnote{A family $\mathcal B$ of subsets of $X$ is a
\emph{basis for a topology} if
$\bigcup_{B\in\mathcal B}B=X$ and, whenever
$B_1,B_2\in\mathcal B$ and $z\in B_1\cap B_2$, some
$B_3\in\mathcal B$ satisfies
$z\in B_3\subseteq B_1\cap B_2$.
The open sets in the generated topology are precisely the
unions of elements of $\mathcal B$.}
for a topology $\tau_n$ on $\mathbb Z^n$.
In this topology, $U_x$ is the smallest open neighborhood of $x$.
Moreover, for every $C\subseteq\mathbb Z^n$, connectedness\footnote{
A subset $C$ of a topological space is \emph{connected} if it
cannot be written as the union of two disjoint, nonempty subsets
that are open in the subspace topology on $C$.
Under this definition, the empty set is connected.}
is understood in the subspace topology, and
\begin{equation}\label{eq:main}
 \boxed{\qquad C\text{ is connected in }(\mathbb Z^n,\tau_n)
 \quad\Longleftrightarrow\quad
 (\forall x,y\in C)\;xE_Cy.\qquad}
\end{equation}
The equivalence includes $C=\varnothing$.
\end{theorem}

\section{Mathematical proof of the main theorem}

\begin{lemma}[Symmetry of adjacency and reversal of parity]
\label{lem:parity}
If $xR_ny$, then $yR_nx$, and $p(x)$ and $p(y)$ have opposite
parities.
\end{lemma}
\begin{proof}
Symmetry follows from $|x_k-y_k|=|y_k-x_k|$.
For every $a\in\mathbb Z$ one has $|a|\equiv a\pmod2$.
Consequently,
\[
 p(x)-p(y)=\sum_{k=1}^{n}(x_k-y_k)
 \equiv\sum_{k=1}^{n}|x_k-y_k|
 =1\pmod2.
\]
Thus $p(x)$ and $p(y)$ have different parities.
\end{proof}

\begin{remark}[Interpretation of parity]
\label{rem:parity}
The map
\[
 \pi:L_n\longrightarrow\mathbb Z/2\mathbb Z,
 \qquad \pi(x)=p(x)\bmod 2,
\]
classifies lattice points as even ($\pi(x)=0$) or odd
($\pi(x)=1$). By Lemma~\ref{lem:parity}, $xR_ny$ implies
$\pi(x)\neq\pi(y)$. This explains definition
\eqref{eq:basic} without introducing further hypotheses.
\end{remark}

\begin{proposition}[Basis and minimal open neighborhoods]
\label{prop:basis}
The family $\mathcal B_n$ satisfies the basis axioms.
Its generated topology is
\begin{equation}\label{eq:topology}
 \tau_n=\bigl\{O\subseteq L_n:\
 \forall x\in O,\ U_x\subseteq O\bigr\}.
\end{equation}
Every $U_x$ is open and is contained in every open set
containing $x$.
\end{proposition}
\begin{proof}
The family on the right-hand side of \eqref{eq:topology}
is a topology: both $\varnothing$ and $L_n$ satisfy the
condition, which is preserved under arbitrary unions and
finite intersections.

We first prove that every $U_x$ is open. If $p(x)$ is odd, then
$U_x=\{x\}$, and for its only point we have
$U_x=\{x\}\subseteq U_x$. If $p(x)$ is even and $z\in U_x$,
there are two possibilities. If $z=x$, then $U_z=U_x$.
If $z\neq x$, definition \eqref{eq:basic} yields $xR_nz$;
by Lemma~\ref{lem:parity}, $p(z)$ is odd and
$U_z=\{z\}\subseteq U_x$. Thus $U_x$ satisfies
the openness condition \eqref{eq:topology} in every case.

Since $x\in U_x$ for every $x\in L_n$, the family
$\mathcal B_n$ covers $L_n$. Moreover, if
$z\in U_x\cap U_y$, the openness of $U_x$ and $U_y$
under condition \eqref{eq:topology} implies
$U_z\subseteq U_x$ and $U_z\subseteq U_y$, hence
\[
 z\in U_z\subseteq U_x\cap U_y.
\]
The basis intersection axiom follows.
Moreover, if $O$ satisfies (6), then
$O=\bigcup_{x\in O}U_x$, since $x\in U_x\subseteq O$
for every $x\in O$. Conversely, every union of members of
$\mathcal B_n$ satisfies (6), by the openness of each $U_x$
already established. Hence the topology generated by
$\mathcal B_n$ is exactly $\tau_n$.

To establish minimality, suppose $x\in U_y$.
If $p(x)$ is even, necessarily $y=x$, so $U_x=U_y$;
if $p(x)$ is odd, then $U_x=\{x\}\subseteq U_y$.
Therefore $U_x\subseteq U_y$ whenever $x\in U_y$.
Since any open set $O$ containing $x$ is a union of
members of $\mathcal B_n$, there is some $U_y$ with
$x\in U_y\subseteq O$. The preceding inclusion gives
$U_x\subseteq O$, as required.
\end{proof}

\begin{lemma}[Accessibility classes]
\label{lem:classes}
Fix $C\subseteq L_n$. Accessibility as defined in
\eqref{eq:path} is an equivalence relation on $C$.
If $A$ is one of its equivalence classes, both $A$ and
$C\setminus A$ are open in the subspace topology on $C$.
\end{lemma}
\begin{proof}
Reflexivity follows from the one-point path;
symmetry follows by reversing a path and using
symmetry of $R_n$; transitivity follows by concatenating
two paths sharing an endpoint.

Let $A$ be an accessibility class. To show that $A$
is open in the subspace topology, fix $x\in A$.
If $y\in U_x\cap C$, then either $y=x$ or $xR_ny$;
in either case $x$ and $y$ are accessible within $C$,
so $y\in A$. Consequently,
\begin{equation}\label{eq:saturation}
 x\in A\quad\Longrightarrow\quad U_x\cap C\subseteq A.
\end{equation}
Each $U_x\cap C$ is open in the subspace topology on $C$.
Since $x\in U_x$ and \eqref{eq:saturation} holds, we obtain
\[
 A=\bigcup_{x\in A}(U_x\cap C).
\]
Hence $A$ is open in $C$. The same argument applies to
every other accessibility class, and their union is
$C\setminus A$. Thus $C\setminus A$ is also open in $C$.
\end{proof}

\begin{proof}[Proof of Theorem~\ref{thm:main}]
Proposition~\ref{prop:basis} establishes the basis and
the minimality of the open neighborhoods.

\emph{First implication.}
Suppose $C$ is connected. If $C=\varnothing$, the
accessibility condition is vacuous. Otherwise choose
$x_0\in C$ and let $A$ be its accessibility class.
By Lemma~\ref{lem:classes}, the disjoint subsets
$A$ and $C\setminus A$ are relatively open and cover
$C$. Since $x_0\in A$, connectedness forces
$C\setminus A=\varnothing$. Thus $A=C$, so every
pair of points in $C$ can be joined by a finite path
within $C$.

\emph{Second implication.}
Suppose every pair of points of $C$ can be joined
by a path within $C$. If $C$ were disconnected,
there would be disjoint, nonempty, relatively open
sets $A$ and $B$ such that $C=A\cup B$.
Choose $a\in A$ and $b\in B$. By assumption,
there is a path $a=z_1,\ldots,z_m=b$ contained in $C$.
Let $j$ be the least index for which $z_j\in B$.
Such an index exists since $z_m=b\in B$, and $j>1$
since $z_1=a\in A$. By minimality of $j$ and the
disjoint partition $C=A\cup B$, we have
\[
 z_{j-1}\in A,\qquad z_j\in B,\qquad z_{j-1}R_nz_j.
\]
By Lemma~\ref{lem:parity}, exactly one endpoint of this
edge has even coordinate sum. Let $e$ be that endpoint
and $o$ the other. Definition \eqref{eq:basic}
implies $o\in U_e$. If $e\in A$, relative openness
of $A$ and minimality of $U_e$ force
$U_e\cap C\subseteq A$, whence $o\in A$, contradicting
$o\in B$. If $e\in B$, interchanging $A$ and $B$
gives the same contradiction. Therefore no separation
of $C$ exists, and $C$ is connected.

For $C=\varnothing$, both properties hold: it cannot
be separated into two nonempty relatively open sets,
and there is no pair of points that must be connected.
This proves \eqref{eq:main}.
\end{proof}

\section{Special case: a topology on the integers}

Recall that $C\subseteq\mathbb Z$ is an
\emph{integer interval}, or a \emph{set of consecutive
integers}, if
\begin{equation}\label{eq:interval}
 a,b\in C,\quad a\leq t\leq b,\quad t\in\mathbb Z
 \quad\Longrightarrow\quad t\in C.
\end{equation}
This includes the empty set, singletons, finite intervals,
and infinite intervals.

\begin{corollary}[Answer to the question on the integers]
\label{cor:integers}
For $\mathbb Z$, consider the family
\begin{equation}\label{eq:zbasic}
 U_m=\begin{cases}
 \{m-1,m,m+1\},&m\text{ is even},\\
 \{m\},&m\text{ is odd}.
 \end{cases}
\end{equation}
There is a topology $\tau_1$ with basis
$\{U_m:m\in\mathbb Z\}$ such that, for every
$C\subseteq\mathbb Z$,
\[
 \boxed{\quad C\text{ is connected in }(\mathbb Z,\tau_1)
 \iff C\text{ is an integer interval}.\quad}
\]
In particular, the answer to the original question is affirmative.
\end{corollary}

\begin{proof}
For $n=1$, adjacency \eqref{eq:adjacency} becomes
$|a-b|=1$, and \eqref{eq:basic} agrees exactly
with \eqref{eq:zbasic}. Theorem~\ref{thm:main}
states that $C$ is connected if and only if every
pair of its points can be joined by a sequence of
integers in $C$ whose successive absolute differences
are one.

If $C$ satisfies \eqref{eq:interval} and $a<b$
belong to $C$, the sequence $a,a+1,\ldots,b$
is such a path lying entirely in $C$.
For $a=b$, a one-point path suffices; for $a>b$,
reverse the order.

Conversely, assume the path property, and let
$a<t<b$ with $a,b\in C$. Take a path
$a=z_1,\ldots,z_m=b$ in $C$ with
$|z_{i+1}-z_i|=1$. Choose the first index $j$
such that $z_j\geq t$. It exists because
$z_m=b>t$, and $j>1$ since $z_1=a<t$.
We have $z_{j-1}<t$ and $|z_j-z_{j-1}|=1$;
therefore $z_j=t$. Thus $t\in C$, proving
\eqref{eq:interval}. The empty and singleton
cases are immediate.
\end{proof}

\section{Methodology of formal verification in Lean 4}

\subsection{Correspondence of mathematical objects}

The reported verification was implemented in Lean~4
with the Mathlib mathematical library \cite{mathlib,lean}.
The lattice $\mathbb Z^n$ was represented by the function type
\texttt{Fin n}\,$\to$\,$\mathbb Z$. On this type,
\texttt{LatticeAdj} (equation \eqref{eq:adjacency})
and \texttt{LatticeEven} (the parity of
\eqref{eq:parity}) were defined.

The topology was encoded using the openness condition
\eqref{eq:topology}, expressed as stability of an open set
under neighbors of a point with even coordinate sum.
Its equivalence with the family $\mathcal B_n$
was formally justified through the statements
\begin{center}
\small
\texttt{basic\_is\_open},\quad
\texttt{open\_iff\_contains\_basic},\\
\texttt{basic\_intersection\_property}.
\end{center}

To model paths inside $C$, adjacency was restricted
to pairs of points both belonging to $C$, and named
\texttt{EdgeIn}. For this relation, the equivalence
closure \texttt{Relation.EqvGen} was used.
Its equivalence to the reflexive-transitive closure
\texttt{Relation.ReflTransGen}, for a symmetric relation,
was verified through the lemma
\begin{center}
\texttt{eqvGen\_iff\_finite\_walk}.
\end{center}

To verify Corollary~\ref{cor:integers} explicitly,
\texttt{IntegerAdj} was defined for the relation $|a-b|=1$.
Integer parity was encoded as \texttt{IntegerEven},
and order convexity \eqref{eq:interval} was represented
by \texttt{IsIntegerInterval}. The basic neighborhoods
in \eqref{eq:zbasic} were verified by the theorems
\begin{center}
\small
\texttt{integer\_basic\_even},\quad
\texttt{integer\_basic\_odd}.
\end{center}

\subsection{Connectedness of the empty set}

It is essential to distinguish \texttt{IsConnected},
which in Mathlib requires a set to be nonempty,
from \texttt{IsPreconnected}, which expresses the
absence of a separation into nonempty relatively
open sets and permits the empty set \cite{mathlibconn}.
The latter notion was used to preserve the
quantification of the main theorem over \emph{all}
subsets $C$.

\subsection{Verification sequence}

The work was divided into four logical levels:
\begin{enumerate}[label=\arabic*.]
\item The topology axioms, basic neighborhoods,
      their openness, and the basis property were formalized.
\item Both implications relating connectedness to paths
      were proved for an abstract symmetric relation
      that reverses parity along adjacent points.
      This stage decomposes the argument into
      accessibility classes, relatively open sets,
      and propagation along a path.
\item The result was specialized to $\mathbb Z^n$,
      deriving symmetry and parity reversal directly
      from \eqref{eq:adjacency}. Thus the abstract
      assumptions are not additional hypotheses
      in the lattice formulation.
\Needspace{5\baselineskip}
\item The characterization on $\mathbb Z$ was formalized separately.
      Both implications between order convexity and accessibility
      by paths of unit steps were proved, as was the equivalence
      between order convexity and topological connectedness.
\end{enumerate}

The results corresponding to the main theorem and the corollary
appear in the source file, respectively, as
\begin{center}
\texttt{lattice\_preconnected\_iff\_consecutive},\\
\texttt{integer\_preconnected\_iff\_interval}.
\end{center}
Their mathematical statements are, respectively, the equivalence
\eqref{eq:main} and the characterization in
Corollary~\ref{cor:integers}.
The compilation with Lean~4.36.0-rc1
and Mathlib completed with exit status zero, no errors or warnings, and
without \texttt{sorry}, \texttt{admit}, or new
axiom declarations in the proof.

\subsection{Correspondence audit}

\begin{table}[ht]
\centering
\caption{Mathematical proof obligations and formalized lemmas.}
\label{tab:audit}
\small
\begin{tabularx}{\linewidth}{@{}>{\raggedright\arraybackslash}p{0.33\linewidth}>{\raggedright\arraybackslash}X@{}}
\toprule
Mathematical obligation & Verification in Lean \\
\midrule
Symmetry of $R_n$ and parity reversal &
\texttt{latticeAdj\_symm}, \texttt{latticeAdj\_flips\_parity} \\
Basis and minimal neighborhood &
\texttt{basic\_is\_open}, \texttt{basic\_le\_open},
\texttt{basic\_intersection\_property} \\
Finite path and accessibility relation &
\texttt{eqvGen\_iff\_finite\_walk} \\
Relatively open classes &
\texttt{pathClass\_saturated}, \texttt{expand\_inter\_eq} \\
Connectedness implies paths &
\texttt{preconnected\_implies\_consecutive} \\
Paths imply connectedness &
\texttt{consecutive\_implies\_preconnected} \\
Complete theorem on $\mathbb Z^n$ &
\texttt{lattice\_preconnected\_iff\_consecutive} \\
Basic neighborhoods at even and odd integers &
\texttt{integer\_basic\_even}, \texttt{integer\_basic\_odd} \\
Integer intervals imply unit-step paths &
\texttt{integer\_interval\_implies\_consecutive} \\
Unit-step paths imply integer intervals &
\texttt{integer\_consecutive\_implies\_interval} \\
Complete corollary on $\mathbb Z$ &
\texttt{integer\_preconnected\_iff\_interval} \\
\bottomrule
\end{tabularx}
\end{table}

The audit distinguishes \emph{formal correctness} from
\emph{historical fidelity}: Lean checks the statement
and its inferences under the formalized definitions,
whereas attribution and presentation choices are
checked against the published formulation of
Problem 5712.

\section{Discussion and conclusions}

Theorem~\ref{thm:main} constructively solves the
problem in all dimensions $n\geq1$.
Parity selects centers whose neighborhoods include
unit lattice steps, while points of opposite parity
have singleton neighborhoods. This asymmetry of
neighborhoods yields an exact correspondence between
the impossibility of separating a subset topologically
and the possibility of traversing it by a finite
sequence of internal adjacencies.

Corollary~\ref{cor:integers} recovers the original
question on $\mathbb Z$. Its characterization of
connected subsets includes finite, infinite, and
singleton intervals, as well as the empty set.
In addition to the verification of the general theorem
on $\mathbb Z^n$, Lean now directly certifies the equivalence
between connectedness and order convexity on $\mathbb Z$,
together with the formulas for the basic neighborhoods
of even and odd integers. These proofs cover both implications
and correctly handle the empty set. No experiments
on finite subsets or second formalization in
another proof assistant were used.

The topological construction and the affirmative
answer are historical results associated with the
Marcus--Wyse problem. The contribution of this work
is to present the generalization as the main theorem,
derive the one-dimensional case rigorously, and
document the correspondence between the mathematical
proofs and their Lean~4 verification.

\end{document}